\documentclass[11pt,english]{article}
\usepackage[utf8]{inputenc}
\usepackage[a4paper, total={6in, 9in}]{geometry}
\usepackage[pdftex]{graphicx}
\usepackage[table]{xcolor}
\usepackage{float}
\usepackage{hhline}
\usepackage{amssymb}
\usepackage{amsmath}
\usepackage{amsthm}
\usepackage{eurosym}
\usepackage{hyperref}
\usepackage{hhline}
\usepackage{soul}
\newtheorem{definition}{Definition}
\newtheorem{proposition}{Proposition}
\newtheorem{theorem}{Theorem}

\newtheorem{lemma}{Lemma}
\newtheorem{corollary}{Corollary}
\newcommand{\commentout}[1]{}

\title{Bounding the Communication Complexity of Fault-Tolerant Common Coin Tossing}
\author{Ivan Geffner \and Joseph Y. Halpern\thanks{Supported in part by NSF grants IIS-1703846 and 
IIS-0911036, ARO grant W911NF-17-1-0592,  MURI grant W911NF-19-1-0217
from the ARO, and a grant from Open Philanthropy.}}
\date{}

\begin{document}

\maketitle

\begin{abstract}
Protocols for tossing a common coin play a key role in the vast majority of implementations of consensus. Even though the common coins in the literature are usually \emph{fair} (they have equal chance of landing heads or tails), we focus on the problem of implementing a \emph{biased} common coin such that the probability of landing heads is $p \in [0,1]$. Even though biased common coins can be implemented using fair common coins, we show that this can require significant inter-party communication.  In fact, we show that there is no bound on the number of messages needed to generate a common coin of bias $p$ in a way that tolerates even one malicious agent, even if we restrict $p$ to an arbitrary infinite subset of $[0,1]$ (e.g., rational numbers of the form $1/2^n$) and assume that the system is synchronous. By way of contrast, if we do not require the protocol to tolerate a faulty agent, we can do this.  Thus, the cause of the message complexity is the requirement of fault tolerance. 
\end{abstract}

%ivan5
\section{Introduction}

%ivan6: rewritten
The idea of a common coin toss was first introduced by
%joe7: using reference in z.bib
%Rabin~\cite{Rabin} as a  core component of his solution to the
Rabin~\cite{Rab} as a  core component of his solution to the
%joe6: you keep saying ``on'' when you mean ``in''
%Byzantine Agreement (or Byzantine Generals) problem on
%joe7
%Byzantine Agreement (or Byzantine Generals) problem in
Byzantine Agreement problem in
\emph{asynchronous systems} (in which messages can be arbitrarily
delayed). In Byzantine Agreement, each of $n$ agents starts with an
input $x_i$, which is assumed to be $i$'s preference. By the end of
the protocol, all honest agents must output the same value $v$, with
the additional constraint that if all honest agents originally had the
same preference $x$, then they must output $x$. Rabin's implementation
%joe8
%assumes that the common coin outcome is given to all agents by a
assumes that the common coin is given to all agents by a
%joe6*: Were Aspnes and Herlihy really the first to do this?  Didn't
%Ben-Or also have a common coin algorithm that didn't not require
%external help?  I think that Aspens and Herlihy may have given the
%first poly-time algorithm for a common coin.  You have to relate your
%work to this literature.  This also makes it clear that your work is
%not an isolated little result.  See, for example
%http://www.cs.yale.edu/homes/aspnes/pinewiki/RandomizedConsensus.html.  
%polynomial time algroithm. 
%ivan7: I think Ben-Or doesn't use a 'common coin', or at least not
%explicitely. I mean, from what I see each agent randomizes its input
%if they can't agree in their first attempt and eventually they all
%randomize the same way, it is sort of a 'common coin' in the big
%picture. I guess it is good to add it. 
%joe7: people understand that Ben-Or implemented a common coin
trusted external entity, but later 
%ivan7: 
%joe7: using reference in z.bib
%Ben-Or~\cite{Ben83} provided a \emph{fault-tolerant} implementation
Ben-Or~\cite{BenOr} provided a fault-tolerant implementation of a
common coin toss 
%joe7
%that did not require any external help and in which agents implicitly
%compute a common coin toss, where \emph{fault tolerant} means that
that did not require external help, 
where \emph{fault tolerant} means that
honest agents still manage to satisfy the problem requirements (in
%joe7
%this case, it is terminating with probability 1, agreeing on the same
%value, and that this value is their preference if all honest agents
%have the same one) even though though some of the agents are
%\emph{faulty} or \emph{malicious} and deviate from the
this case, terminating with probability 1, agreeing on the same
value, with that value being the preference of the honest agents if
all honest agents 
have the same preference) even if the faulty agents 
deviate in arbitrary ways from the
protocol. Ben-Or's protocol for $n$ agents tolerates up to $n/2$
\emph{fail-stop} faults (where agents can deviate by stop sending
%joe7
%messages after some point), and up to $n/5$ byzantine faults (where
messages after some point), and up to $n/5$ Byzantine faults (where
agents can deviate in any way), and the expected number of messages
%joe7
%was exponential in $n$. A few years later, Bracha~\cite{Bracha87}
is exponential in $n$. Bracha~\cite{Bracha87}
improved Ben-Or's solution by implementing a more efficient common
%joe7
%coin, and provided a protocol that tolerates up to $n/3$ byzantine
%faults, which was already proven to be
coin; he provided a protocol that tolerates up to $n/3$ Byzantine
faults, which was known to be
optimal~\cite{BrachaT83}. 
%ivan10*: this assumes shared memory, I don't know if there is a better citation
%\commentout{
%ivan11: replaced by what you have in your consensus paper
\commentout{
%joe1: I still think we should reference it.  You may want to do a
%google search for other solutions.
Bracha's solution still required an
%joe11
%exponential number of messages in expectation, but Aspnes and 
exponential number of messages in expectation.  Aspnes and 
Herlihy~\cite{AspnesH1990consensus} provided 
a polynomial time fault-tolerant
%joe7
%common coin that tolerated up to $n/3$ faults.
%joe11
%common coin that tolerates up to $n/3$ faults.
%}
common coin that tolerates up to $n/3$ faults, but their approach used
shared memory.
}
%joe12
%Bracha's solution still required an
Bracha's solution still requires an
%joe11
%exponential number of messages in expectation, but Aspnes and 
exponential number of messages in expectation, but Feldman and
Micali~\cite{FM88} provided a constant expected time fault-tolerant
common coin that tolerates up to $n/3$ faults. 
%ivan7: out
\commentout{
 Bracha 
Aspnes and
%joe6: can you send me the reference if we use it?  
%Herlihy~\cite{AspnesH1990consensus} came up with a
%\emph{fault-tolerant} common coin implementation without any external
%help, where \emph{fault-tolerant} means that honest agents still
%joe6: can you send me the reference if we use it?
Herlihy~\cite{AspnesH1990consensus} provided
\emph{fault-tolerant} common-coin implementation that did not require
%ivan7: I think this is important
external
help, where \emph{fault-tolerant} means that honest agents still
manage to agree on a random common bit even though some of the agents
deviate from the protocol. 
}
Ever since, nearly all implementations of
%joe8
%Byzantine Agreement rely on the implementation of common coins.
Byzantine Agreement rely on common coins. 
%joe6*: ran on paragraph.  What other examples are there?  You need to
%give some references here; you can't just say this
%Aside from Byzantine Agreement, which is the most important example,
%common coin tosses are essential in many other protocols that require
%coordinated randomization.
%ivan7: actually I'd take this out since it is stated for biased coins right below, it seems redundant.
\commentout{
Being able to toss a common coin also plays a significant role in
protocols other than Byzantine agreement and consensus, such as ...
}

%joe6*: added paragraph break here.  You should also confirm that in
%the aplications where the common coin is used, they often need a
%biased coin.
%In the literature, usually the common coin
%implementation is meant to be \emph{fair}, which means that the
%implemented coin has equal probability of landing heads (1) and tails
%(0). The main reason behind this is that, as shown later in this
%joe8
%The common coin implemented by the protocols is typically a fair coin,
The common coin that is implemented is
typically a fair coin, 
%joe7
%that has an equal likelihood of landing heads and tails.  In
which has an equal likelihood of landing heads and tails.  In
%joe7*: Do all proof of stake protocols use weighted leader election?
%Am I right that the probability of an agent being elected is
%proportional to the stake?  In any case, a reference is needed here
%protocols, we often need a biased coin
%%ivan7:
%as for example in a weighted random leader election (as in \emph{proof
%  of stake} blockchain protocols), in a distributed lottery, or more
%in general any distributed protocol that requires uneven global
%randomization.  Do you have an example of a paper that uses
%distributed lotteries?  If so, that should be cited.
%ivan8: these protocols never talk explicitly about it, but they are
%equivalent problems and thus implcitly implement a common coin. 
many protocols, a biased coin is needed.  For example, in \emph{proof
  of stake} protocols, we want the probability of
an agent being elected leader to be proportional to the agent's stake
%ivan8:
%\cite{??};
\cite{NK18};
 the same is true in distributed lotteries
 %ivan8:
 \cite{GrumbachR17}.
%joe6*: we need to stress that we require fault tolerance.
%section, any $p$-\emph{biased} coin, which is a coin that has
%probability $p$ of landing heads, can be implemented by tossing
%\emph{fair} coins repeatedly. However, as we show in this paper, this
%is at the expense of a much larger number of messages. In fact, we
%show that even in systems with pairwise reliable, private and
Of course, a biased coin can be implemented by tossing a fair coin
repeatedly, but, as we show in the main result of this paper, doing so
will typically require agents to send many messages if we want the
%joe7
%common coin to be \emph{fault tolerant}, that is, if we require that,
common coin to be fault tolerant in the sense that, 
no matter what malicious agents do, honest agents will agree on the
outcome of the coin toss.  
In fact, we
show that even in systems with pairwise reliable, private, and
authenticated synchronous channels and with an additional broadcast
channel, there exist values of $p \in [0,1]$ such that the amount of
rounds of communication required to implement a fault-tolerant
%joe7: ``Even more'' is bad English in this context
%$p$-biased coin is arbitrarily high. Even more, we show that this
$p$-biased coin is arbitrarily high. Indeed, we show that this
property holds even if we restrict the possible values of $p$ to any
%joe7
%infinite subset $S \subseteq [0,1]$. This means that, for instance, if
infinite subset $S \subseteq [0,1]$ and want to tolerate only one
malicious player.
This means that, for instance, if
we only consider rational values of $p$, or only values of the form
%joe6
%$1/2^k$, some of those values require an arbitrarily high number of
$1/2^k$, there is no bound on the number of messages required to
%joe7
%generate a $p$-biased coin.
generate a common $p$-biased coin.
%joe7*: there's a wole area called communication complexity that
%studies problems like this
Put another way, the communication complexity of a protocol that takes
as input a number $p$ in some infinite subset $S$ of $[0,1]$ and results
in a (fault-tolerant) common $p$-biased 
coin must have unbounded 
%ivan10
worst-case
communication complexity for all choices of $S$.
%ivan10: added
If we consider the expected communication complexity instead of the worst case, we show in Section~\ref{sec:def} that there exists a family of protocols
that generate common $p$-biased coins
%ivan10:
with probability $1$ for all $p \in [0,1]$
%joe11
%and have expected message complexity of 2, but do not have a uniformly
and have an expected message complexity of 2, but do not have a uniformly
bounded message complexity. 
%ivan10: added this, our result is similar in flavor to FLP
This means that the problem of implementing $p$-biased coins 
%joe11
%is somewhat similar to that of implementing distributed consensus, in the
has the same lower-bound behavior as that of distributed consensus, in the
sense that there is no protocol that implements the desired
functionality while tolerating a single fault and guaranteeing
%joe11
%termination (\cite{FLP} in the case of distributed consensus), but
termination, but
there exist such protocols that terminate with probability 1
%joe11
(see \cite{ADH08,FLP} for the results in the case of consensus).

It may seem that the problem is due to the ``complexity'' of
$p$.  Perhaps it requires $O(k)$ messages to generate a coin
of bias $1/2^k$.  However, as we show in 
%ivan10
%see 
%ivan7:
%Section~\ref{...},
Section~\ref{sec:def}, 
 the
 difficulty arises from the requirement of fault tolerance.
%joe8*
%ivan10: I would take this out
\commentout{
 and that the bound holds in all executions of the protocol (as
 opposed to just holding with high probability).
}
 There is a
simple protocol that requires only $n$ messages for $n$ agents to generate a
$p$-biased common coin if we do not require fault tolerance for all $p
\in [0,1]$.  (Generating this coin might require some local
%joe7
%computation, but very
computation, but agents need to send very 
%ivan7
few
%joe7
%messages are needed.)  Moreover, as we also
%joe8
%messages.)  Moreover, as we also
messages.) 
%ivan10: moved
\commentout{
 Moreover, we also provide a trivial family of protocols
%show, if we do not care about message complexity, for all $p \in
%[0,1]$ there is a trivial protocol $\pi^p$ with 
%ivan7:
%bounded 
%message
%complexity
%joe8
%that generates a common $p$-biased coin,
that generate  common $p$-biased coin for all $p \in [0,1]$
%joe7
%joe8*
%but these protocols do not have a uniformly bounded message complexity.
that do not have a uniformly bounded message complexity,
but all have expected message complexity of 2.
}

%ivan11: this was originally why we stumbled upon this problem to begin with.
%joe14
%This result has also consequences in Algorithmic Game Theory. Given a
%game $\Gamma$, implementing solution concepts such as Nash or
%correlated equilibria with cheap talk has been a central subject of study
%ivan11: probably I'm missing citations:
%(see \cite{AH03,Bp03,Gerardi04,GM05,ADGH06,ADGH19}). Roughly speaking,
These results turn out to be relevant to work in \emph{implementing
  mediators} using cheap talk (i.e., 
getting the same outcome that can be obtained using a trusted mediator
by just having the agents communicate with each other) 
%ivan11: probably I'm missing citations:
%joe14: yes, there are many, many more papers (Francois Forges in
%perhaps one to add).  I'm not sure I wold call this algorithmic game
%theory.  All the papers except the ADGH papers are by economists.  I
%took out ours, since it seems a little self serving to reference that
%when there's so much other stuff.
%(see, e.g., \cite{AH03,Bp03,Gerardi04,GM05,ADGH06,ADGH19}). Roughly speaking,
%cheap talk extensions of $\Gamma$ consist of two phases, one in which
%players are allowed to freely communicate, and then another in which
%they must make a move in $\Gamma$. In the literature it is somewhat
%assumed that it is common knowledge when the communication phase ends,
%and thus that (a) every player makes a move in $\Gamma$
(see, e.g., \cite{AH03,Bp03,F90,Gerardi04,GM05,ADGH06, ADH07, abraham2019implementing, geffner2021security, geffner2023lower, geffner2023communication}).  
In these
implementations, roughly speaking, players first talk to each other
%(this is called the \emph{communication phase})
and then at some point must make a move in an underlying game $\Gamma$.
In the literature, it is implicitly
assumed that players move simultaneously.  
%joe14: readers won't  know what a punishment strategy is.
%simultaneously, and (b) players cannot communicate after they play an
%action. Knowing when the communication phase ends is of critical
%importance in certain situations, for example when there exists a
%punishment strategy (\cite{Bp03,ADGH06,ADGH19}) since players do not
%know if they have to punish or not until the very end of the game. The
%end of the communication phase being common knowledge is indeed a
%natural assumption if the communication protocol $\vec{\sigma}$ is
%\emph{uniformly finite} (i.e. it always terminates in at most $N$
%action.
It is typically straightforward to assure this,
since the implementations almost always run in a bounded
number of rounds; that is, there is some $N$ such that, all executions
of the algorithm terminate in at most $N$ rounds.  Thus, players can
just wait until round $N$, and then make a move in the underlying game.
%The end of the communication phase being common knowledge is indeed a
%natural assumption if the communication protocol $\vec{\sigma}$ is
%%\emph{uniformly finite} (i.e., for some $N$, it always terminates in
%%at most $N$ rounds), since players can just wait until the
%last round to play their action. However, if $\vec{\sigma}$ is not
%uniformly finite, the fact that players know when all other players
%terminate places an additional assumption on the protocol settings.
%joe14*: Ivan, I'm totally lost.  what do the results of this paper 
The move made is typically drawn according to some distribution (e.g.,
the distribution determined by the Nash equilibrium of the game); \emph{a priori}, 
%joe14
%that involve an infinite number of different probabilities have an
the algorithm must be able to deal with 
an infinite number of possible probabilities.  As the results of this
paper show, this means that the implementation cannot run in a bounded
number of rounds.
%joe14*
But if the algorithms are unbounded,
it becomes nontrivial for players to coordinate on a commonly-agreed
%ivan12: time?
%upon tiee 
time
to make their moves.  
%ivan12: Is this actually true? I'd cut it since we have already
%mentioned the implicit assumption (simultaneous moves).
%%joe15: If not, how are they impementing the coin tosses?  In any
%%case, if this isn't true, why are we even talking about all this work?
%I don't see where you talk about this earlier, and the paragarph
%doesn't flow if you cut it, so I reinstated it.  If you point out
%where we discussed this issue, then I'd still keep the sentence, but
%I would rewrite it a bit.
%\commentout{
The papers on implementing mediators
get around this problem by implicitly assuming that players have 
access to arbitrarily biased coins
%ivan13: I think Gerardi and some others do this instead
%joe16
%or that players have no computational constraints (e.g. that players can
or that players have no computational constraints (i.e., that players can
operate with arbitrary real numbers and encode them in their messages).
%}
Thus, the implicit assumption
made in these papers is actually hiding a lot of complexity.
%This includes, for example, correlated
%equilibria in games with an infinite number of actions, or solution
%concepts involving the simulation of a trusted mediator with
%non-finite randomization. 

%joe6: ``In'' not ``over'' 
%Over the next section we formalize all these notions introduced in
%this section and introduce our main result.  
In the next section we formalize all the notions introduced in
this section, give the simple protocols mentioned 
%ivan10* now "above" also refers to two paragraphs above
above, and state 
our main result.
%joe6*: Added; this is important
The details of the proof are given in 
%ivan7:
%Section~\ref{...}.
Section~\ref{sec:proof}.
Perhaps surprisingly, the key technical result that we need for the proof is
a nontrivial theorem from algebraic geometry that says that 
%ivan7: this is not true (it can have infinite critical points with
%the same image), also it doesn't need to be multilinear, we need the
%multilinear property this for our finiteness argument. 
%a multilinear polynomial can have only finitely many  
the set of \emph{critical points} of a 
%joe7: need to explain image
%polynomial in several variables has a finite image, where $x$ is a
polynomial $Q$ in several variables has a finite image; that is,
$\{Q(x): x \mbox{ is a critical point of $Q$}\}$ is finite,  where $x$ is a
critical point of $Q$  if all partial derivatives of $Q$ at $x$ are 0.

\section{Basic Definitions and Results}\label{sec:def}

%joe3
%In this section we provide the main definitions and results. We begin
In this section, we provide the main definitions and state the main
results. We begin 
%ivan7:
%of
with
 a formalization of 
%ivan8: 
% $p$-consensus:
%joe8
 % implementing
 common $p$-biased coins: 

\begin{definition}
%ivan8:
%[$p$-consensus]
Given a subset $T$ of malicious agents and a strategy $\vec{\tau}_T$
for agents in $T$, let $\vec{\pi}(T, \vec{\tau}_T)$ be the output
distribution of protocol $\vec{\pi}$ when players in $T$ play
$\vec{\tau}_T$ instead of $\vec{\pi}_T$. A joint protocol $\vec{\pi}$
%joe3
%is a $t$-resilient implementation of $p$-consensus for $n$ agents if
is a \emph{$t$-resilient implementation of
%ivan8:
%$p$-consensus
a 
  common $p$-biased coin 
 for $n$
  agents} if
  %joe3:
%  for all subsets $T$ of malicious players with $|T| \le t$ and all
  in all systems of $n$ agents,
for all subsets $T$ of malicious players with $|T| \le t$, and all
strategies $\vec{\tau}_T$, 
%joe3*: there are many problems here: First of all, the output is a
%tuple, not a bitstring.  Second, why should there be *exactly* n-|T|
%ones.   If that's what you require, the players in T can easily
%deviate by outputtng 
%$\vec{\pi}_{-T}(T, \vec{\tau}_T)$ is a
%random variable that takes value $1^{n - |T|}$ (which is $n - |T|$
%ones) with probability $p$ and $0^{n - |T|}$ with probability $1-p$.
%ivan7*: this is not the correct definition, since honest players may
%disagree even if there are n-|T| ones. Also it doesn't say what
%happens with probability 1-p. In the previous definition we were
%looking at \pi_{-T} (which is the output of players not in T,
%malicious players cannot break the definition by outputting something
%else. 
\commentout{
The probability according to $\vec{\pi}(T, \vec{\tau}_T)$ that there
are at least $n-|T|$ 1s in the output is $p$.
}
the random variable $\vec{\pi}_{-T}(T, \vec{\tau}_T)$ (which is the
output  of $\vec{\pi}(T, \vec{\tau}_T)$ restricted to players not in
$T$) is of the form $(1, 1, \ldots, 1)$ with probability $p$ and
$(0,0, \ldots, 0)$ with probability $1-p$. 
\end{definition}

%joe8
In short, a protocol $\vec{\pi}$ is a $t$-resilient implementation of a
%joe3
%$p$-consensus if regardless of what a coalition of up to $t$ malicious
%ivan8:
%$p$-consensus regardless, if of what a coalition of up to $t$ malicious
common $p$-biased coin if, regardless of what a coalition of up to $t$
malicious  
players 
%ivan10:
%do,
does,
 honest players still manage to agree on a random bit that
takes value $1$ with probability $p$ and $0$ with probability
$1-p$. Our main result is that no family of $1$-resilient protocols
%joe8
%that implements
that implement 
%ivan8:
%$p$-consensus
$p$-biased coins
 for several $p$ values in $[0,1]$ can be
implemented with a uniform bound $N$ on the number of messages sent by
honest agents. In order to simplify notation, given a protocol
$\vec{\pi}$ we define an \emph{adversary} to be a pair of the form
$(T, \vec{\tau}_T)$, where $T$ is the subset of malicious players and
$\vec{\tau}_T$ is the strategy they play (instead of
$\vec{\pi}_T$). We also define $\vec{\pi}(A)$ as the output
distribution of running $\vec{\pi}$ with adversary $A$. 

\begin{definition} A set $P$ of joint protocols is \emph{$N$-uniformly
    bit-bounded} if 
%joe19
  %there is a bound $N$ such that all
  all 
honest players send at most $N$ bits in all histories of all joint protocols
%joe1
%ivan7:
$\vec{\pi} \in P$.
%joe9: added
$P$ is \emph{uniformly bit-bounded} if 
%ivan10
%if
%joe11
%it $N$-bit-bounded for some $N$.
it is $N$-bit-bounded for some $N$.
\end{definition}

%joe1
%$\vec{\pi} \in P$. Our main result is proving that there is no uniform
%ivan7: this goes inside the definition environment
%$\vec{\pi} \in P$.
%joe7: we already said this a few sentences ago
%Our main result shows that there is no uniform
%upper bound on the number of messages required to implement common
%joe1
%coin tossing with different probabilities while tolerating a single
%coin tossing with different probabilities while tolerating even a single
%malicious player:
We can now state our main result:

%joe2: let's call it a theorem
%\begin{proposition}\label{proposition:main}
\begin{theorem}\label{proposition:main}
%ivan5: reformulated:
\commentout{
%joe1: rewrote; it's too fuzzy
%  There exists no $1$-resilient implementation for $p$-Consensus such
%that the expected number of bits sent by honest players in
%$\vec{\pi}^p$ is uniformly bounded for all $p \in \mathbb{R} \cap
  %[0,1]$.  
There is no 
%ivan4: important
%joe2: wrong place
%$1$-resilient
uniformly bit-bounded family $\vec{\pi}^p$ of
joint protocols (parameterized by $p \in 
%joe2
%[0,1]$) such that $\vec{\pi}^p$ implements 
%ivan8:
%$p$-consensus
a common $p$-biased coin 
 for $p \in [0,1]$.
[0,1]$) such that $\vec{\pi}^p$ 
%ivan5:
\commentout{
implements
$p$-consensus and is 1-resilient for all $p \in [0,1]$. 
}
is a $1$-resilient implementation of $p$-consensus for $p \in [0,1]$.
%joe2$
%\end{proposition}
}
%joe3
%If $n > 1$, for every infinite set $S \subseteq [0,1]$, there is no
For all $n > 1$ and all infinite sets $S \subseteq [0,1]$, there is no
%joe3
%uniformly bit-bounded family $\vec{\pi}^p$ of 
%joint protocols  for $n$ agents(parameterized by $p \in
%S$) such that $\vec{\pi}^p$
uniformly bit-bounded family $P_S = \{\pi^p: p \in S\}$ of
joint protocols 
such that $\vec{\pi}^p$ 
is a $1$-resilient implementation of a
%ivan8:
%$p$-consensus
common $p$-biased coin 
 for $n$ agents for
all $\pi^p \in P_S$.
\end{theorem}

%ivan5:
%joe3*: As I said before, this people will find this a competely
%uninteresting corollary
%As an immediate consequence of Theorem~\ref{proposition:main} we have
%the following corollary: 
%
%\begin{corollary}
%There is no uniformly bit-bounded family $\vec{\pi}^p$ of
%joint protocols such that $\vec{\pi}^p$ 
%is a $1$-resilient implementation of $p$-consensus for $p \in [0,1]$.
%\end{corollary}

%ivan5:
%\section{Proof}

%ivan4:
%joe2: moved this below, after the examples
%We will prove Proposition~\ref{proposition:main} in a synchronous

%ivan4:
%joe2
%We next provide, under this model, an $(n-1)$-resilient family
Before proving Theorem~\ref{proposition:main}, 
%ivan5:
%joe3: undid change.  The hype about why it's important should be in
%the introduction
we show, as we claimed in the introduction,
%it is important to note 
that the two
assumptions (1-resilience and uniformly bit-boundedness) are necessary.
We first provide an $(n-1)$-resilient family
%ivan5: changing all [0,1] for S
$\{\vec{\pi}^p\}_{p \in S}$ of joint protocols that 
%ivan5
%joe8: undid
implement
%implements
%ivan8:
%$p$-consensus
common $p$-biased coins
%joe2
(where $n$ is the number of players) 
that is 
%joe2
%\emph{not} uniformly bounded. Given $p \in [0,1)$, suppose
  %  $0.a_1a_2\ldots$ is the binary representation of $p$ (i.e. $a_k$ is
%  the $k$ digit in base 2 after the decimal point). At each round $r
%  \ge 1$ each player $i$ broadcasts a random bit $b_{i,r}$ set to 0 or
  %  1 with probability $1/2$ each. Then, if $r > 1$, each player $i$
%    computes $b_{r-1} := (\sum_{j \in [n]} b_{j,r}) \bmod 2$, in which
  %  $b_{j, r-1}$ is set to $0$ if $j$ didn't broadcast a bit in round
%    $r-1$. If $b_{r-1} < a_{r-1}$, $i$ outputs $0$ and terminates. Else
%  if $b_{r-1} > a_{r-1}$, $i$ outputs $1$ and terminates. Note that
\emph{not} uniformly bit-bounded.
%ivan5:
%With
If $1 \in S$, with
 joint protocol $\vec{\pi}^1$, all players simply output 
%ivan5: p is the probability to output 1 now
%0
1
 and terminate.
For joint $\vec{\pi}^p$ with $p \in S \setminus \{1\}$, we proceed as follows.  Suppose
that $0.a_1a_2\ldots$ is the binary representation of $p$ (i.e., $a_k$ is
%joe2
%the $k$th digit in base 2 after the decimal point). At each round $r
the $k$th digit after the decimal point in the binary representation
of $p$). At each round $r
  \ge 1$, each player $i$ broadcasts a uniformly distributed random
  bit $b_{i,r}$.
%joe9: moved up from below
(Note that even though we are assuming the existence of a broadcast
%joe9
%  channel, Bracha~\cite{Bracha87} provided a deterministic broadcast
  channel, this is without loss of generality, since
  Bracha~\cite{Bracha87} provides a deterministic broadcast 
  protocol that requires finitely many rounds of
communication.) 
  If $r > 1$, each player $i$
  computes $b_{r-1} := (\sum_{j \in [n]} b_{j,r}) \bmod 2$ (where we take
  $b_{j, r-1} = 0$ if $j$ did not broadcast a bit in round
  $r-1$). If $b_{r-1} < a_{r-1}$, then $i$ outputs 
%ivan5:  
%  $0$
$1$
   and terminates;
  if $b_{r-1} > a_{r-1}$, then $i$ outputs 
%ivan5:
  %$1$
  $0$
   and terminates;
%joe2: added
  if $b_{r-1} = a_{r-1}$ and $i$ has not yet terminated, then $i$
  continues to the next round.
  Note that
%joe1
  %  since all the bits are being broadcasted, all honest players agree
%  on the same value. Moreover, as long as a single player is honest,
%  each bit $b_r$ is equal to 0 and to 1 with probability $1/2$
%  each. Thus, we can easily show by induction that the probability
%  each. Thus, we can easily show by induction that the probability
  since all the bits are broadcast, all honest players agree
  on their value. Moreover, since at least one player is honest,
  each bit $b_r$ is uniformly distributed.
%ivan5:
%each. 
%joe3: What exactly is the inductive claim.  This needs to be made
%more precise
%ivan7:
\commentout{
  We can easily show by induction that the probability
  that honest players agree on 
%ivan5:
  %0
  1
   is $\sum_{r \ge 1} a_r 2^{-r}$, which
  is exactly $p$. 
  }
  We can easily show by induction on the number of communication
  rounds that the probability 
  that honest players agree on 
  $1$ by the end of the $r$th round of communication is $\sum_{k =
    1}^r a_k 2^{-k}$, the probability that they agree on $0$ is
  $\sum_{k = 1}^r (1-a_k) 2^{-k}$, and the probability that they have
  not agreed yet is $2^{-r}$. This means that the probability that the
%joe7
%  agents eventually terminate is $1$, and that they agree on $1$ is
  agents eventually terminate is $1$, 
%ivan9:  
 %and
  that the probability that
  they agree on $1$ is 
  $\sum_{r \ge 1} a_r 2^{-r}$, which is exactly $p$,
  %ivan8:
%joe9
  %  and that the expected number of messages sent by each agent is 2.
%joe10
%  and that each agent sends 2 messages in expected number of messages.
    and that each agent sends 2 messages in expectation.
%joe9: this is overkill
%  Note, however,
%  that this protocol requires an unbounded amount of messages for all
%  $p$, even though in some cases it can be modified in such a way that
%  it only requires a finite amount. For instance, if $p$ is of the
%  form $a/2^b$, after the $b$th round of communication agents can
%  simply output $0$ since all the remaining digits are zeroes. 
Note, however, that this family of protocols has no uniform bound on
the number of messages sent for all $p \in [0,1]$.
  %ivan8:
  %joe9:
  %Therefore we have the following result:
  The argument above proves the following result:
\begin{theorem}\label{thm:expected-bounded}
  %joe9*: You need to talk about n
%  There exists a family of joint protocols $\{\vec{\pi}^p : p \in
  For all $n$, 
  there exists a family of joint protocols $\{\vec{\pi}^p : p \in
[0,1]\}$ such that $\vec{\pi}^p$
%ivan9: rewritten
\commentout{
 is a $t$-resilient implementation of
%joe9
%a common $p$-biased coin for all $t$, and such that $\vec{\pi}^p$ has
%a finite expected message complexity. 
 is a $t$-resilient implementation of 
a common $p$-biased coin 
%ivan9:
for $n$ agents, 
for all $t<n$,
}
%joe10
%is an implementation of a common $p$-biased coin for $n$ agents which
is an implementation of a common $p$-biased coin for $n$ agents that
is $t$-resilient for all $t < n$, 
 and 
%joe10*
 % $2n$ messages are sent in $\vec{\pi}^p$ for all $p \in [0,1]$.
such that,  for all $p \in [0,1]$, the expected number of messages
sent in $\vec{\pi}^p$ is $2n$.
\end{theorem}

%joe9: moved up the discussion of Braha that was here before.

%joe2
%  Note that there exists a trivial family $\{\vec{\pi}^p\}_{p \in
  %  [0,1]}$ of protocols that implements $p$ consensus for all $p \in
%[0,1]$ in which players send at most $n$ bits: Player 1 generates a
%joe3*: This assumes implicitly that players have access to a device
%that lets the generate a bit with probaility p for arbitrary p.
%People won't like this assmption, and it has to be made explicit.
  %You shold say that the lower bound holds even if players can do that.
  There 
%ivan7:
also  
  exists a trivial family $\{\vec{\pi}^p\}_{p \in
        S}$ of $n$-uniformly bit-bounded joint protocols (where $n$ is
  the number of players) that
  $\vec{\pi}^p$ implements 
  %ivan8:
  %$p$-consensus 
  common $p$-biased coins 
  for all $p \in 
S$: In 
%ivan5:
%$\vec{\pi}6p$,
$\vec{\pi}^p$,
 player 1 generates a
bit $b$ that is 
%ivan5:
%$0$ with probability $p$ and $1$ with probability
$1$ with probability $p$ and $0$ with probability
$1-p$
 and broadcasts $b$
%ivan7: 
%joe7
%(note that we can assume that $i$ can generate such bit without loss
 % of generality since even if $i$ can only generate bits with equal
(note that we can assume that $i$ can generate such a bit without loss
 of generality since even if $i$ can generate only bits with an equal
 probability of being 0 and 1, it can run the protocol described
%joe7
 % before locally);
  above locally to generate a bit with an arbitrary bias);  
  each player outputs whatever is sent by
player $1$. Of course, this protocol is not $1$-resilient: if player
$1$ was malicious it could send a bit with a different probability or
no bit at all. This, along with 
%ivan9
%previous example,
%joe9
%the example provided for
Theorem~\ref{thm:expected-bounded},
 shows that
%joe2
%Proposition~\ref{proposition:main} is no longer true if we drop either
%ivan5:
%Proposition~\ref{proposition:main}
Theorem~\ref{proposition:main} 
 does not hold if we drop either
the $1$-resilience or the uniformly bit-bounded condition.  

%ivan5:
\section{Proof of Theorem~\ref{proposition:main}}\label{sec:proof}

%ivan4:
%The
%ivan5: out
\commentout{
We now prove Theorem~\ref{proposition:main}.  We do so in a synchronous
model in which players have broadcast channels
%joe2: added
(in addition to point-to-point communication).
%ivan4: I don't know if we should add this
%Note that if we prove Proposition~\ref{proposition:main} in this case,
%it is also true for asynchronous systems and when players don't have
%broadcast channels. 
%joe2
%The result immediately follows true for asynchronous systems and if
The result immediately follows for asynchronous systems and if
there are no broadcast channels. 
}

%joe2
%To prove Proposition~\ref{proposition:main}, the
% first step is showing that we can assume without loss of
The first step in the proof is showing that we can assume without loss of
generality that agents use a \emph{basic protocol}, in which 
%joe1*: I don't know what it means that a message has a fixed number of
%bits.  Every message has a fixed number of bits.  I wrote what I
%suspect you meant
%all honest agents send a single message with a fixed amount of bits $k$
there is some $k$ such that, at the first round, all 
honest agents broadcast a single 
%ivan4: it is easier this way:
%$k$-bit message,
message of at most $k$ bits,
 and in the second
round, the agents  output
%joe
%a non-deterministic function of all bits sent in round 1.
a (possibly randomized) function of the bits received in round 1.
%joe3: removed paragraph break
%ivan5: should we say it was introduced by us? I think we don't need
%both directions for this case, however this is more general 
%joe3*: We should at least point to the use of t-bimsimlation in our
%other paper, although there the definition was more complicated since
%it involved asynchrony.  I also don't understand what you mean when
%you say ``we don't need both directions''.  What don't we need both
%directions for?
%For this purpose we need the following definition:
%ivan7: we only need that all the outcomes in \pi' can be achieved in
%\pi, but not the other way around. However this way we highlight how
%strong is this construction since it goes both ways. 
To make this precise, we need the following
%joe7
%definition, which was introduced by Geffner and Halpern~\cite{GH18}:
definition, due to Geffner and Halpern~\cite{geffner2021security}: 

\begin{definition}
%joe3
%  A protocol $\vec{\pi}$ $t$-bisimulates a protocol $\vec{\pi}'$ if
  A protocol $\vec{\pi}$ $t$-bisimulates protocol $\vec{\pi}'$ if
  \begin{itemize}
%joe3: you have to first define an adversary as a pair (T,\tau)
%  \item[(a)] For all adversaries $A = (T, \vec{\tau}_T)$ there exists an
      \item[(a)] for all adversaries $A = (T, \vec{\tau}_T)$
%ivan10:
with $|T| \le t$      
       there exists an
  adversary $A' = (T, \vec{\tau}'_{T})$ such that $\vec{\pi}(A)$ and
%joe3
  %  $\vec{\pi}'(A')$ are identically distributed.
%\item[(b)] For all adversaries $A' = (T, \vec{\tau}'_T)$ there exists
  $\vec{\pi}'(A')$ are identically distributed; and 
\item[(b)] for all adversaries $A' = (T, \vec{\tau}'_T)$ 
%ivan10:
with $|T| \le t$ 
there exists
  an adversary $A = (T, \vec{\tau}_{T})$ such that $\vec{\pi}(A)$ and
  $\vec{\pi}'(A')$ are identically distributed. 
\end{itemize}
\end{definition}

%ivan5: actually it also means that the adversary learns the same
%information in both scenarios, however this is trivial in this case
%because there are no inputs. 
Intuitively, a protocol $\vec{\pi}$ $t$-bisimulates another protocol
$\vec{\pi}'$ if they produce the same outputs regardless of what the
adversary does. The fact that we can assume without loss of generality
that agents use a basic protocol is a
%joe3
%direct consequence of the following proposition.
follows immediately from the following proposition. 

\begin{proposition}\label{prop:implementation}
%ivan5: usually in math we write $i \in I$, where I is an arbitrary
%set of index. I left 'i' as it is since I don't know if this is
  %standard notation.
  %joe5: this is fine if you write I; without it, it's horrible
  %notation.  You also want to say that \tau_i is basic.  Finally, i
  %is overloaded, since you also use it for an agent.
  %If $\{\vec{\pi}^i\}_i$ is a uniformly bit-bounded family of joint
%protocols, there exists a family of uniformly bit-bounded joint
%protocols $\{\vec{\tau}^i\}_i$ such that $\vec{\tau}^i$
%$t$-bisimulates $\vec{\pi}^i$ for all $t$ and $i$. 
%ivan8: refactor
\commentout{
  If $\{\vec{\pi}^j: j \in J\}$ is a uniformly bit-bounded set of joint
  protocols (for some index set $J$), then there exists a uniformly
  bit-bounded set 
  $\{\vec{\tau}^j: j\in J\}$ of basic protocols such that $\vec{\tau}^j$
%joe4
  %  $t$-bisimulates $\vec{\pi}^i$ for all $t$ and $i$.
    $t$-bisimulates $\vec{\pi}^j$ for all $t$ and $j \in J$. 
    }
      If $\{\vec{\pi}^j: j \in J\}$ is a uniformly bit-bounded set of joint
  protocols (for some index set $J$), then there exists a set 
  $\{\vec{\tau}^j: j\in J\}$ of basic protocols and a constant $k$ such that
\begin{itemize}
\item [(a)]   $\vec{\tau}^j$ $t$-bisimulates $\vec{\pi}^j$ for all $t$
%joe9
  %  and $j \in J$.
%\item [(b)] In every history of every joint protocol $\vec{\tau}^j$
  and $j \in J$;  and
\item [(b)] in every history of every joint protocol $\vec{\tau}^j$
  with $j \in J$, each message from each honest agent contains exactly
  $k$ bits. 
\end{itemize}
\end{proposition}

%ivan8:
%joe9: already defined (actually, not quite, because you mean exactly
%k here, but we can worry about that later.
%By the nature of Proposition~\ref{prop:implementation}, it is useful
%joe8
%to define a $k$-uniformly bit bounded family of joint protocols as one
%in which the size of the messages of honest players in each of those
%protocols is of exactly $k$ bits. 
%joe9
%Note that,
Note that
%ivan8
 %in particular,
 in Proposition~\ref{prop:implementation},
  if $\vec{\pi}^j$ is a $t$-resilient
%joe3
%implementation of $p$-consensus for some $t$ and $p$, so is
%$\vec{\tau}^j$. In particular we have the following corollary: 
implementation of 
%ivan8:
%$p$-consensus
a common $p$-biased coin 
 for some $t$ and $p$, then so is
%joe9
 $\vec{\tau}^j$,
%ivan8: 2x in particular
%In particular, 
 %Therefore,
 so
we have the following corollary: 

\begin{corollary}\label{corollary1}
  If there exists a set $S \subseteq [0,1]$ and a uniformly bit-bounded set
  $\{\vec{\pi}^p: p \in S\}$ of joint protocols
  such that $\vec{\pi}^p$ is a $1$-resilient implementation of a common $p$-biased coin for each $p \in S$,
%joe5
%  there exists a uniformly bit-bounded family of basic joint protocols
%  $\{\vec{\tau}^p\}_{p  \in \mathbb{R} \cap  [0,1]}$
%ivan8: rewritten
%joe9
%  then there exists $k \in \mathbb{N}$ and a $k$-uniformly bit-bounded set
  then for some $k \in \mathbb{N}$, there exists a uniformly
  bit-bounded set 
  $\{\vec{\tau}^p: p \in S\}$
 of basic joint protocols
 such that $\vec{\tau}^p$ is a $1$-resilient implementation of a
 %joe8
 common 
 $p$-biased coin for each $p \in S$,
 %joe9:
 and each agent sends exactly $k$ bits in $\tau^p$.
\end{corollary}

%ivan5:
%\begin{proof}
\begin{proof}[Proof of Proposition~\ref{prop:implementation}]
%ivan4: rewritten from scratch
%joe2
%  Given a protocol $\vec{\pi}$, let $\vec{\tau}^{\vec{\pi}}$ be the
  Given a joint protocol $\vec{\pi}$, let $\tau_i^{\vec{\pi}}$ be the
  protocol where player $i$ broadcasts a message in round 1 that
%joe2
%contains all messages that $i$ would send with $\pi_i$ for all local
%histories with at most $N$ bits.  Then, at round $2$, each player $i$
contains all the messages that $i$ would send with $\pi_i$ in 
each local history $h$ with at most $N$ bits (for instance, if 
%ivan10: n -> m
%$h_i^1,\ldots, h_i^n$
$h_i^1,\ldots, h_i^m$
 are all possible histories of $i$ with at most $N$
%joe7
%bits, $i$ can send a message of the form $(msg_1, \ldots, msg_k)$,
bits, then $i$ can send a message of the form 
%ivan10: k -> m
%$(msg_1, \ldots, msg_k)$,
$(msg_1, \ldots, msg_m)$,
%ivan10: i might send more than one message
\commentout{
where $msg_j$ is the message that $i$ would send with history $h^j_i$,
with the recipient and the channel encoded in the message itself).
}
where $msg_j$ contains all the messages that $i$ would send with history
$h^j_i$, with the recipients and the channels encoded in the message itself).
In
round $2$, $\tau_i^{\vec{\pi}}$   
%joe2
%simulates what local history each player would have in $\vec{\pi}$ if
%all players played according to what they sent in their message at
%round 1. This is done inductively as follows: given the simulated
%joe3
%simulates what local history each player $j$ would have in $\vec{\pi}$ if
computes what local history each player $j$ would have in $\vec{\pi}$ if
all players played according to the protocol implicitly described in
their round-1 message 
%joe3
%round 1. This is done inductively as follows: given the simulated
round 1. This is done inductively as follows: Suppose that $i$ has
computed the
local history $\vec{h}^r$ that players would have by the end of round
%joe3
%$r$,  $i$ computes $\vec{h}^{r+1}$ by appending all messages that
$r$.  Agent $i$ then computes $\vec{h}^{r+1}$ by appending all messages that
each 
player $j$ would have sent given $h^r_j$, according to the message $j$
sent at round 1. If $j$ sent a message with an incorrect format at
%joe2
%round 1, $i$ simulates that $j$ sends no messages at all. Player $i$
round 1, then $i$ assumes that $j$ sends no messages at all
according to its protocol. Player $i$
%joe2
%computes $\vec{h}^r$ this way until reaching a terminating state (in
%joe5
%computes $\vec{h}^r$ this way until it reaches a terminating state
computes $\vec{h}^r$ in this way until it reaches a terminating state 
%joe3
%(i.e., a state in which no player sends a message from there on).
(i.e., a state in which no player sends a message from then on).
%joe2
%Clearly, by the end of
%round 2 all honest players would have simulated the same history.
Clearly, by the end of 
round 2, all honest players will have simulated the same history. 

%ivan8:
%joe9
%We now show part (a).
We now prove part (a). 
It is
easy to check that, by construction, for all adversaries $A = (T,
\vec{\rho}_T)$, 
$\vec{\pi}(A)$ and $\vec{\tau}^{\vec{\pi}}(T, (\vec{\tau}^{\vec{\rho}})_T)$
%joe7: removed blank line
%
are identically distributed for all inputs $\vec{x}$. Moreover, given
an adversary $A = (T, \vec{\rho}_T)$ in $\vec{\tau}^{\vec{\pi}}$,
%joe2
%consider an adversary $A = (T, \vec{\rho}')$ that does the following:
consider an adversary $A = (T, \vec{\rho}'_T)$ that does the following:
each player $i \in T$ computes which message $msg_i$ it would send at
%joe2*: shouldn't this be according to \pi_i, not \tau^\pi
%round $1$ if it played $\vec{\rho}$. If $msg_i$ has a correct format
%(according to $\vec{\tau}^{\vec{\pi}}$), it acts in $\vec{\pi}$
%exactly as described in $msg_i$, otherwise it sends no messages at all
%during the whole protocol. Again, by construction
%ivan5: we are going the other way around, given an adversary in the
%basic protocol, we want to construct an adversary in the original
%protocol. In this case, they would compute what would they have sent
%in the basic one and play according to it. 
%joe5
%round $1$ if it ran $\rho_i$. If the format $msg_i$ is correct
%joe3: what does it mean that the format is correct according to \pi?
%This ahs t obe defined
%ivan7: we already defined how is a 'correct' message in \pi
round $1$ if it ran $\rho_i$. If the format of $msg_i$ is correct 
(according to $\pi_i$), when running $\rho_i'$, $i$  acts 
%joe3: What does it mean to act exactly as described in a message?
%Messages in general to describe how to act.
exactly as described in $msg_i$
%ivan7:
%(recall that $msg_i$ should describe all messages that $i$ would send
(recall that $msg_i$ is supposed to describe all messages that $i$ would send
in every possible history); 
 otherwise, $i$ sends no messages at
all.
%ivan7: redundant
%when $\rho'_i$.  
Again, by construction,
%ivan7:
%$(\vec{\tau}^{\vec{\pi}}, A, \vec{x})$ and $(\vec{\pi}, A', \vec{x})$
$\vec{\tau}^{\vec{\pi}}(A)$ and $\vec{\pi}(A')$
are identically distributed for all inputs $\vec{x}$. This shows that
$\vec{\tau}^{\vec{\pi}}$ $t$-bisimulates $\vec{\pi}$ for all $t$.  

Finally, note that the length of the message that each honest player
sends in round 1 with $\vec{\tau}^{\vec{\pi}}$ is bounded by a
function of $N$, and thus if  
%joe5
%$\{\vec{\pi}^p\}_{p \in [0,1]}$ is uniformly bit-bounded, then so is
%$\{\vec{\tau}^{\vec{\pi}^p}\}_{p \in [0,1]}$.
$\{\vec{\pi}^j: j\in J\}$ is uniformly bit-bounded, then so is
$\{\vec{\tau}^{\pi^j}: j \in J \}$. 
%ivan8:
Part (b) follows from the fact that if $\{\vec{\pi}^j : j \in J\}$ is
a uniformly bit-bounded set of basic joint protocols such that the
size (in bits) of all messages from honest agents is bounded by some
constant $k$, then there exists a set of basic joint protocols
$\{\vec{\pi}^j_* : j \in J\}$ such that $\vec{\pi}^j_*$
$t$-bisimulates $\vec{\pi}^j$ for all $t$ and all $j \in J$, and all
messages of honest agents in every protocol $\vec{\pi}^j_*$ have
exactly $k + \log k$ bits. The construction of $\vec{\pi}^j_*$ is as
follows: if agent $i$ would send message $msg$ at round 1 with
$\vec{\pi}^j$, it sends a message $msg'$ of exactly $k + \log k$ bits
in $\vec{\pi}^j_*$, where the first $\log k$ bits of $msg'$ encode the
number of bits of $msg$ (in binary), and the last $k$ bits of $msg'$
contain $msg$ followed by a sequence of zeroes. It is easy to check
that each message $msg'$ in $\vec{\pi}^j_*$ defines a unique message
$msg$ in $\vec{\pi}^j$ (and vice-versa), and thus that $\vec{\pi}^j_*$
$t$-bisimulates $\vec{\pi}^j$ for all $t$ and all $j \in J$ as
desired.  
\end{proof}

%ivan8: I'll reestructure from here on.
Corollary~\ref{corollary1} implies that, to prove
Theorem~\ref{proposition:main}, it suffices to show that there are no
families of basic joint protocols that implement $p$-consensus for
infinite values of $p$ and such that the message size (in bits) is
fixed.
%joe8
%The proof revolves around a standard result in Algebraic
%Geometry\footnote{Cite here}. 
The key step proof involves a standard result in algebraic
geometry.

%ivan8: if we are doing a definition for critical points, I might add critical values as well.
\begin{definition}
Given a polynomial $P$ in several variables, a point $\vec{x}$ is a \emph{critical point} of $P$ if all partial derivatives of $P$ at $\vec{x}$ are 0. A value $v$ is a \emph{critical value} of $P$ if it is the image of a critical point under $P$ (i.e. if there is a critical point $\vec{x}$ such that $P(\vec{x}) = v$).
\end{definition}

\begin{theorem}\label{thm:poly} 
  %joe8: here is where you put the citation
  \cite[Exercise 4.9]{Coste02}
All polynomials $P$ in several variables have finitely
many critical values.
\end{theorem}

%ivan8:
The main idea for proving Theorem~\ref{proposition:main} is to show
that we can associate a polynomial $Q^{\vec{\pi}}$ in several
%joe8
%variables to each basic joint protocol $\vec{\pi}$ such that all
%messages of honest agents have fixed size $k$. Moreover, we show that
variables with each basic joint protocol $\vec{\pi}$ such that all
messages of honest agents have a fixed size $k$. Moreover, 
if $\vec{\pi}$ is a 1-resilient implementation of a $p$-biased coin,
%joe8
%then $p$ is a critical value of $Q$. Theorem~\ref{proposition:main}
then $p$ is a critical value of $Q^\pi$. Theorem~\ref{proposition:main}
follows from the fact that if a family of basic joint protocols is
%joe8
%$k$-uniformly bit bounded, then there are only finitely many protocols
%$Q^{\vec{\pi}}$ in our correspondence, and therefore one of them must 
$k$-uniformly bit-bounded, then there are only finitely many such 
%ivan11
%protocols
polynomials
$Q^{\vec{\pi}}$, so one of them must 
have infinitely many different critical values, which contradicts
Theorem~\ref{thm:poly}. Before constructing $Q^{\vec{\pi}}$,
%joe8
%as an intermediate step
we construct a simpler polynomial $P^{\vec{\pi}}$
such that there exists a transformation $\gamma$ that maps
%joe8
%$P^{\vec{\pi}}$ to $Q^{\vec{\pi}}$ (i.e. such that $P^{\vec{\pi}}
$P^{\vec{\pi}}$ to $Q^{\vec{\pi}}$ (i.e., such that $P^{\vec{\pi}}
\circ \gamma = Q^{\vec{\pi}}$). We show that if $\vec{\pi}$ is a basic
$1$-resilient implementation of a $p$-biased coin such that all
messages of honest agents have fixed size $k$, then there exists a
point $\vec{x}$ such that $P^{\vec{\pi}}(\vec{x}) = p$ and another
linear map $\phi$ such that $\phi(\vec{x})$ is a critical point of $Q$
and $Q(\phi(\vec{x})) = p$. 

%joe8: ``goes as follows'' is not good English
%The construction of $P^{\vec{\pi}}$ goes as follows. Suppose that
The construction of $P^{\vec{\pi}}$ proceeds as follows. Suppose that
$\vec{\pi}$ is a basic 1-resilient implementation of a $p$-biased coin
for some $p \in [0,1]$ such that the number of bits sent by honest
agents is exactly $k$. 
%ivan8: added what was before, replacing pi^p for pi
Let $f^{\vec{\pi}}(s_1, \ldots, s_n)$ be the 
decision function 
of honest players, that is, the function that, given $s_1, \ldots, s_n
\in \{0,1\}^{k}$, computes whether the honest players output 0 or 1 as
a function 
of the histories $s_1, \ldots, s_n$ of all players. 
Note that 
if $n > 1$, 
since all honest players must compute the same output, their output has to be a deterministic function of $(s_1, \ldots, s_n)$, which means that $f^{\vec{\pi}}$ is deterministic.
Let
$x_s^i$ be the probability that player $i$ sends message $s$ at round
1 if $\vec{\pi}$ is run. Then the probability that
players output $1$ is
$$P^{\vec{\pi}}(\vec{x}) = \sum_{s_1, \ldots, s_n \in
  \{0,1\}^{k}} x_{s_1}^1 \ldots x_{s_n}^n f^{\vec{\pi}}(s_1, \ldots,
s_n),$$
where $\vec{x}$ is the vector whose components are the $2^kn$
probabilities $x_s^i$.

Note that $P^{\vec{\pi}}(\vec{x})$ is 
a multilinear
 polynomial in
$2^k n$ variables
with all 
%ivan10:
%or
of
its coefficients in $\{0,1\}$.
 If $\vec{\pi}^p$ is $1$-resilient
 and $n > 1$, 
 then
 $P^{\vec{\pi}}(\vec{x})$ must remain constant regardless 
%joe8
 % of how the round-1 message sent by the
  of how the message sent in the first round by the 
malicious player
is chosen.
Since the distribution of honest players' outputs 
cannot be affected by deviations of malicious players, even if
some malicious player 
decides to send its first message following another distribution in which
each message $s$ has probability $y_s^i$ to be sent (instead of $x_s^i$), it must
be the case that,  
 for each $i \in [n]$, $P^{\vec{\pi}}(\vec{x})$ must remain
constant when the values $\{x_s^i\}_{s \in \{0,1\}^k}$ are replaced by
certain other values $\{y_s^i\}_{s \in \{0,1\}^k}$ such 
that, for all $i \in [n]$, 
$\sum_{s \in
  \{0,1\}^k} y_s^i = 1$.
\begin{definition}
%joe2: no need to say ``We say that'' in a definition.  Also, are
%these standard definitions?  If so, you should say that (and point
%toa  reference).  Also, you should say somewhere that \matbb{C}
%represents the complex numbers
%  Let $k$ and $n$ be two positive integers. We say that a point $\vec{x}
%\in \mathbb{C}^{kn}$ is $(k,n)$-probabilistic if $x_i \ge 0$ for all $i$
%ivan5:  
  %Given positive integers $k$ and $n$,  
  %$\vec{x} \in \mathbb{C}^{kn}$ is 
  Given positive integers $\ell$ and $n$,  
 $\vec{x} \in \mathbb{C}^{\ell n}$ is 
%ivan5:  
  %\emph{$(k,n)$-probabilistic}
  \emph{$(\ell,n)$-probabilistic}
   if,
  %joe2
%$x_i \ge 0$ for all $i$
%and, for all $i \in [n]$, $\sum_{j = 1}^{k} x_{(i-1)k + j} = 1$. 
  for all $i \in [n]$,
%ivan5: 
  %$x_i \ge 0$ and $\sum_{j = 1}^{k} x_{(i-1)k + j} = 1$. 
  $x_i \ge 0$ and $\sum_{j = 1}^{\ell} x_{(i-1)\ell + j} = 1$. 
\end{definition}

%ivan8: intuition+
%joe8
%Intuitively, a point $\vec{x}$ is $(\ell, n)$ probabilistic if it is
Intuitively, a point $\vec{x}$ is $(\ell, n)$-probabilistic if it is
the concatenation of $n$ sub-arrays of length $\ell$ that could be
%joe8
%probability distributions (i.e. that all values are positive and that
probability distributions (i.e., all $\ell$ values are positive and 
their sum is 1). Note that these sub-arrays are actually $(\ell,
1)$-probabilistic points. Our previous argument shows that if
$\vec{\pi}$ is a $1$-resilient basic implementation of a $p$-biased
%joe8
%coin such that the messages of honest agents have exactly $k$ bits,
%then there exists a $(2^k, n)$ probabilistic point $\vec{x}$ such that
common coin in which the messages of honest agents have exactly $k$ bits,
then there exists a $(2^k, n)$-probabilistic point $\vec{x}$ such that
$P^{\vec{\pi}}(\vec{x})$ remains constant even if we replace one of
%joe8
%its sub-arrays of length $2^k$ by any other $(2^k, 1)$-probabilistic
%point. This is formalized as follows:
its sub-arrays of length $2^k$ by another $(2^k, 1)$-probabilistic
point. This observation motivates the following definition:

\begin{definition}\label{def:stable}
%joe2
%  Fix $k$ and $n$, and let $P$ be a polynomial in $kn$ variables. We say
  %that a point $\vec{x}$ is $(p,k,n)$-stable in $P$ if:
  %ivan5
  \commentout{
    Given positive integers $k$ and $n$ and a polynomial $P$ in $kn$
    variables, a point $\vec{x}$ is \emph{$(p,k,n)$-stable in $P$} if
  \begin{itemize}
%joe2
%\item [(a)] $\vec{x}$ is $(k,n)$-probabilistic.
  %\item [(b)] Let $\vec{z}_i = (x_{(i-1)k + 1}, \ldots, x_{ik})$. For
  \item [(a)] $\vec{x}$ is $(k,n)$-probabilistic;
\item [(b)] for
  all $i \in [n]$ and all $(k,1)$-probabilistic points
  $\vec{y}$,
%joe2
if $\vec{z}_i = (x_{(i-1)k + 1}, \ldots, x_{ik})$, then
  $$P(\vec{z}_1, \vec{z}_2, \ldots, \vec{z}_{i-1}, \vec{y},
  \vec{z}_{i+1}, \ldots, \vec{z}_n) = p.$$ 
\end{itemize}
}
  Given positive integers $\ell$ and $n$ and a polynomial $P$ in $\ell n$
%joe3*: please make it ``stable point of P'' everywere, so that it's
%parallel to ``critical point of Q''
%  variables, a point $\vec{x}$ is \emph{$(p,\ell,n)$-stable in $P$} if
  variables, a point $\vec{x}$ is 
%ivan7
a  
  \emph{$(p,\ell,n)$-stable point of $P$} if
  \begin{itemize}
%joe2
%\item [(a)] $\vec{x}$ is $(k,n)$-probabilistic.
  %\item [(b)] Let $\vec{z}_i = (x_{(i-1)k + 1}, \ldots, x_{ik})$. For
  \item [(a)] $\vec{x}$ is $(\ell,n)$-probabilistic;
\item [(b)] for
  all $i \in [n]$ and all $(\ell,1)$-probabilistic points
  $\vec{y}$,
%joe2
if $\vec{z}_i = (x_{(i-1)\ell + 1}, \ldots, x_{i\ell})$, then
  $$P(\vec{z}_1, \vec{z}_2, \ldots, \vec{z}_{i-1}, \vec{y},
  \vec{z}_{i+1}, \ldots, \vec{z}_n) = p.$$ 
\end{itemize}
\end{definition}

%ivan8: given before
\commentout{
%joe3*: Ivan, you need to give *much* more intuition here, both for
%the definitions, and in the ``argument'' (there was no ``previous
%argument'', just a claim that a property held).
The previous argument can be summarized in the following statement:
}

\begin{lemma}\label{lemma:stable}
%joe2*: we don't talk about a single protocol being uniformly
%bit-bounded, but a family of joint protocols being uniformly
%bit-bounded.  I don't think we need to say bit-bounded, since you're
%already saying that players send at most k bits
%  If $\vec{\pi}^p$ is a $1$-resilient uniformly bit-bounded basic
  %protocol that implements $p$-consensus in which players send at most
    If $\vec{\pi}^p$ is a $1$-resilient basic
protocol that implements $p$-consensus in which players send 
%ivan8:
%at most
exactly
$k$ bits
%ivan8:
in their first message, 
 there exists a
%ivan8:
multilinear polynomial $P^{\vec{\pi}^p}$ in $2^kn$ variables 
with coefficients in $\{0,1\}$ and a point $\vec{x}$ such that $\vec{x}$ is a
%ivan7:
%transformation
%ivan5:
%$(1-p, 2^k, n)$-stable
$(p, 2^k, n)$-stable
 point in
$P^{\vec{\pi}^p}$. 
\end{lemma}

%joe8
%We now show how to construct $Q^{\vec{\pi}}$. More precisely, we show
%the following:
The next proposition now shows how to construct the desired polynomial
$Q^{\vec{\pi}}$.  

%ivan8: I only included one of the transforms, the other one is
%implicit in the construction of Q. 
\begin{proposition}\label{prop:linear-maps}
Given a polynomial $P \in \mathbb{C}[X]$ in $\ell n$ variables, there
exist a polynomial $Q$ in $(\ell - 1)n$ variables and a linear
transform 
%ivan10:
%$\gamma^\ell_n$ 
$\phi^\ell_n$ 
such that if $\vec{x}$ is a $(p, \ell,
n)$-stable point of $P$, then $\phi^\ell_n(\vec{x})$ is a critical
point of $Q$ such that $Q(\phi^\ell_n(\vec{x})) = p$. 
\end{proposition}

\begin{proof}
  Intuitively, a $(\ell,1)$-probabilistic point $\vec{x}$ can be
  parametrized by its first $\ell-1$ components (since $x_\ell = 1 - \sum_{i =
    1}^{\ell-1} x_i$); similarly, a $(\ell,n)$-probabilistic point
can be parametrized by 
$(\ell-1)n$ of its
components
using a function $\gamma_n^\ell$.
 The idea of the proof is to 
%ivan8: 
%define a polynomial $Q$ that
define $Q$ as the polynomial that
acts on each $\vec{x} \in \mathbb{C}^{(\ell-1)n}$ in the same way that
$P$ would act on $\gamma_n^\ell(\vec{x})$, and check that
$(p,\ell,n)$-stable points in $P$ correspond to critical points 
%joe3: why use a new word (image) when you can explain it simply.
%in $Q$ with image $p$.
$y$ of $Q$ such that $Q(y) = p$.
%ivan8:

%joe3
%Suppose $P$ has an infinite number of $(p,\ell,n)$-stable points.
Suppose that $P$ has an infinite number of $(p,\ell,n)$-stable points.
 Let
%joe3*: What's C?  The complex numbers?  Our domains are real, not
%complex?  Should this be R
 $\gamma^\ell: \mathbb{C}^{\ell-1} \rightarrow \mathbb{C}^{\ell}$ be the linear
transformation defined by $\gamma^\ell(x_1, \ldots, x_{\ell-1}) = (x_1, \ldots,
x_{\ell-1}, 1 - \sum_{i = 1}^{\ell-1} x_i)$, and let $\gamma^\ell_n :
\mathbb{C}^{(\ell-1)n} \rightarrow \mathbb{C}^{\ell n}$ be the transformation
defined by $\gamma^\ell_n (x_1, \ldots, x_{(\ell-1)n}) =
(\gamma^\ell(\vec{z}_1), \ldots, \gamma^\ell(\vec{z}_n))$, where
$\vec{z}_i :=
(x_{(i-1)(\ell-1) + 1},\allowbreak \ldots, x_{i(\ell-1)})$. In addition, let $\phi^\ell:
\mathbb{C}^\ell \rightarrow \mathbb{C}^{\ell-1}$ be the transformation defined by
$\phi^\ell(x_1, \ldots, x_\ell) = (x_1, \ldots, x_{\ell-1})$, and $\phi_n^\ell:
\mathbb{C}^{\ell n} \rightarrow \mathbb{C}^{(\ell-1)n}$ be defined as
$\phi_n^\ell(x_1, \ldots, x_{\ell n}) = (\phi^\ell(\vec{z}_1'), \ldots, \phi^\ell(\vec{z}_n'))$,
where $\vec{z}_i' = (x_{(i-1)\ell + 1}, \ldots, x_{i\ell})$.
Intuitively, $\phi_n^\ell$ is the right inverse of $\gamma_n^\ell$ on
$(\ell,n)$-probabilistic points (i.e. $(\gamma_n^\ell \circ \phi_n^\ell)(\vec{x})
= \vec{x}$ for all $(\ell,n)$-probabilistic points $\vec{x}$).  

Consider the polynomial $Q = P \circ \gamma_n^\ell$ in $(\ell-1)n$
variables.
%joe8*: Ivan, the next two sentences don't belong here.  Just prove
%the theorem; no point in getting the contradiction yet.
%We show below that if $\vec{x}$ is a $(p,\ell,n)$-stable point
%in $P$, then $\phi_n^\ell(\vec{x})$ is a critical point of $Q$ and
%$Q(\phi_n^\ell(\vec{x})) = p$.
%Since $S$ (the set of probabilities $p$ for which we want a 
%%joe8
%%$p$-consensus protocol) is infninite, this contradicts
%$p$-consensus protocol) is infinite, this contradicts 
%Theorem~\ref{thm:poly}, since it shows that
%$Q$ has infinitely many critical points.
Note that $(\gamma_n^\ell \circ \phi_n^\ell) (\vec{x}) = \vec{x}$ for all
$(\ell,n)$-probabilistic points $\vec{x}$, which means that if $\vec{x}$ is
$(p,\ell,n)$-stable in $P$, then 
$Q(\phi_n^\ell(\vec{x})) = P(\vec{x}) = p$.
It
remains to show that $\phi_n^\ell(\vec{x})$ is a critical point in $Q$. 

To see this, let $\vec{z}_i = (x_{(i-1)\ell + 1}, \ldots, x_{i\ell})$, and
consider the polynomial $$Q_i^{\vec{x}}(\vec{y}) :=
Q(\phi_n^\ell(\vec{z}_1), \ldots, \phi_n^\ell(\vec{z}_{i-1}), \vec{y},
\phi_n^\ell(\vec{z}_{i+1}), \ldots, \phi_n^\ell(\vec{z}_n)).$$ 
Let $\Delta^{k-1}$ be the standard simplex in $\mathbb{R}^{\ell-1}$ defined by the equations $y_j \ge 0$ for all $j$ and $\sum_{j = 1}^{\ell-1} y_j \le 1$. Note that if $\vec{y} \in \Delta^{\ell-1}$,
 then $\gamma^\ell(\vec{y})$ is $(\ell,1)$-probabilistic and therefore,
by part (b) of Definition~\ref{def:stable}, we have that
 $P(\vec{z}_1, \ldots, \vec{z}_{i-1}, \gamma^\ell(\vec{y}),
 \vec{z}_{i+1}, \ldots, \vec{z}_n) = p$.  
 This implies that $Q_i^{\vec{x}}(\vec{y}) = p$ for all $\vec{y} \in \Delta^{\ell-1}$, which means that $Q_i^{\vec{x}}$ is constant. In particular, by taking $\vec{y} = \phi^\ell(\vec{z}_i)$,
  we have that $\frac{\partial Q_i^{\vec{x}}}{\partial y_j}(\phi^\ell(\vec{z}_i)) = 0$ for all $j \in [\ell-1]$, which implies that $\frac{\partial Q}{\partial y_{i(\ell-1) + j}} (\phi_n^\ell(\vec{x})) = 0$. 
    Since this argument applies to all $i \in [n]$ and all $j
  \in [\ell-1]$, it follows that $\phi_n^\ell(\vec{x})$ is a critical point
  of $Q$. 
\end{proof}

%ivan8:
Now we can prove Theorem~\ref{proposition:main}.

\begin{proof}[Proof of Theorem~\ref{proposition:main}]
Suppose that $S$ is an infinite subset of $[0,1]$ and suppose there
exists a family $\{\vec{\pi}^p : p \in S\}$ of $k$-uniformly
bit-bounded joint protocols such that each $\vec{\pi}^p$ is a
%joe8
%$1$-resilient implementation of a $p$-biased coin.
%Then, by Lemma~\ref{lemma:stable} there exists a family
$1$-resilient implementation of a common $p$-biased coin.
Then by Lemma~\ref{lemma:stable}, there exists a family
$\{P^{\vec{\pi}^p} : p 
\in S\}$ of multilinear polynomials in $2^kn$ variables with
coefficients in $\{0,1\}$ such that, for each $p \in S$, there exists
a point $\vec{x}^p$ such that $\vec{x}^p$ is a $(p, 2^k, n)$-stable
point of $P^{\vec{\pi}^p}$. Since there are only finitely many
multilinear polynomials in $2^kn$ variables with binary coefficients,
there must exist $q \in S$ such that  $P^{\vec{\pi}^q}$ has infinitely
many $(p, 2^k, n)$-stable points. By
Proposition~\ref{prop:linear-maps}, there exists a polynomial $Q$ such
that each $(p, 2^k, n)$-stable point in $P^{\vec{\pi}^q}$ can be
mapped to a critical points in $Q$ with image $p$. This implies that
the polynomial $Q$ provided by Proposition~\ref{prop:linear-maps} has
infinitely many different critical values, which contradicts
Theorem~\ref{thm:poly}. 
\end{proof}

\section{Conclusion}

We have shown that there is no uniform bound on the number of messages
%joe7
%required to agree on random coin tosses. This implies that certain
required to achieve fault-tolerant $p$-consensus, even if we want to
tolerate only one failure. This implies that certain
%joe7
%kind of protocols that require coordinated randomization
kind of protocols, like weighted leader election, that require
coordinated randomization with probability $p$ for a potentially
infinite set of values $p$ 
%that require coordinated randomization 
%cannot be implemented with a uniform bound either.
must have unbounded communication complexity.
%joe7
%ivan8: I'm dumb, the simple protocol that we give at the beginning
%only uses 2 messages in expectation... 
\commentout{
Our result considers only worst-case communication complexity.  
We
believe that it can be extended to expected communication complexity.
%We expect this result to
%still hold true even if we consider the expected number of messages
%instead of the worst-case number of messages: 

\begin{conjecture}
If $n > 1$, for every infinite set $S \subseteq [0,1]$ and all $N \in \mathbb{N}$, there is no family $\vec{\pi}^p$ of
joint protocols  for $n$ agents such that $\vec{\pi}^p$ 
is a $1$-resilient implementation of $p$-consensus for $p \in S$ and each honest agent $i$ in each protocol $\vec{\pi}^p$ sends less than $N$ messages in expectation.
\end{conjecture}
}

%ivan8: 
We also showed that the cause of the message complexity is the
%joe8
%requirement of fault-tolerance. In fact, there are both protocols that
requirement of fault-tolerance combined with the requirement that the
protocol be uniformly bounded in all executions. In fact, there are
both protocols that 
implement arbitrary $p$-biased coins but are not fault tolerant, and
protocols that implement arbitrary $p$-biased coins while tolerating
any number of malicious players, but in which the number of messages
%joe9
%sent by honest players in the worst-case scenario is arbitrarily
sent by each honest player in the worst-case scenario is arbitrarily
%joe8
%high.
high, although it is only 2 in expectation. 

%joe7
%joe8*: You don't want to leave solving this conjecture to future
%work!  instead, you want to show that it follows from what we did!
%We leave solving this conjecture to future work.
%joe1
%\bibliography{joe,game1}
%ivan11:
%\bibliography{z,joe}
\bibliographystyle{plain}
\bibliography{z,joe}

\end{document}